\documentclass[conference,letterpaper]{IEEEtran}
\usepackage[cmex10]{amsmath} 
\usepackage{graphicx} 
\usepackage{amsthm}
\usepackage{amssymb}
\usepackage{amsfonts}
\usepackage{amsmath}
\usepackage{amsthm}
\usepackage{amssymb}
\usepackage{dsfont}

\newtheorem{definition}{Definition}
\newtheorem{lemma}{Lemma}
\newtheorem{theorem}{Theorem}

\newtheorem{remark}{Remark}
\newtheorem{assumption}{Assumption}

\newcommand{\tti}{\rightarrow\infty}
\newcommand{\graph}{\mathcal{G}}

\newcommand{\prob}[1]{\mathbb{P}\left(#1\right)}

\newcommand{\plim}[0]{p\mbox{-}\lim}
\newcommand{\plimsup}[0]{p\mbox{-}\limsup}
\newcommand{\pliminf}[0]{p\mbox{-}\liminf}
\newcommand{\Var}{\text{Var}}

\begin{document}
\title{Entropy and Distributed Source Coding of Connected Soft Random Geometric Graphs} 


\author{%
  \IEEEauthorblockN{Oliver Baker and Carl P. Dettmann}
  \IEEEauthorblockA{School of Mathematics \\
                    University of Bristol\\
                    Bristol, United Kingdom\\
                    Email: \{o.baker, carl.dettmann\}@bristol.ac.uk}
}

\maketitle

\begin{abstract}
    We consider the distributed compression of Soft Random Geometric Graphs (SRGGs) above the connectivity threshold. We establish the Slepian-Wolf rate region for the SRGG in the setting where there are a finite number of encoders compressing sections of the graph independently. To do so, we prove novel limit theorems and asymptotic equipartition properties for the SRGG and its entropy, which allow us to use random binning techniques for distributed compression.
\end{abstract}

\section{Introduction}

In the modern world, devices are more connected than ever. This results in networks becoming larger, increasing the need for research on compression of networks. In particular, in areas of application such as wireless networking, random models of large networks with a spatial embedding have gained popularity \cite{haenggi_2012}. One popular model for wireless networks is the \textit{Soft Random Geometric Graph} (SRGG). In this model, nodes are connected at random according to a function of their spatial distance \cite{penrose2016connectivity}. 

Despite their numerous applications, the information theory of these random graphs is not well understood. In recent years there have been a growing number of works that have concentrated on this model or similar ones such as graphons. In \cite{janson2010graphons} a limit theorem for the entropy of $W$ random graphs (which include dense SRGGs as a special case) was proven. Since then, there have been a number of works attempting to bound or quantify the entropy, such as \cite{coon2018entropy, baker2026entropy, coon2018conditional, vippathalla2026entropy}. Also, the compression of random graphs has gained popularity recently. Models such as sparse Erd\H os-R\'enyi graphs \cite{choi_structural_entropy}, stochastic block models \cite{Abbe2016GraphClusters, martin_rate_distortion_sbm, bhatt2023universal} and preferential attachment graphs \cite{luczak_pag} have all been considered. Compression of the SRGG (or generalisations) has been studied in \cite{delgosha_universal_compression, delgosha2022universal, vippathalla2024lossy}. In works such as \cite{choi_structural_entropy, mihai2021structural}, the goal is to compress the \textit{structure} of the graph, whereas our aim is to compress the \textit{labelled} graph. Reference \cite{Paton2022labelled_entropy} gives a detailed treatment of the differences between these regimes. 

Crucially, all of these works focus on compression in the centralised setting, where the encoder sees the entire graph. However, in many applications, it is perhaps more realistic to assume that nodes, or blocks of nodes only have access to partial information about the network. In \cite{delgosha2021distributed}, distributed compression of Erd\H os-R\'enyi and configuration model graphs is studied, but distributed compression of the SRGG is as yet unstudied. This regime is of interest in SRGGs because the spatial nature of the graph can induce natural limits on how much of the graph can be observed. We consider the regime where we have a fixed, arbitrary number of encoders which can access different sections of the graph. Beyond graph compression, there are examples of situations where a central controller needs (at least partial) topological knowledge from distributed network locations. For example, sink nodes or central controllers in wireless sensor networks that make routing, scheduling, interference management or topology control decisions \cite{tunca2013distributed}.

Our other source of novelty is the specific regime in which we study the SRGG. In the above works, when the SRGG (or $W$ random graph) is made sparse (the expected degree of a node is sublinear in the total number of nodes), a fraction of the edges of a dense ensemble are removed randomly. However, this places less emphasis on the geometry, since the length scale of edges does not scale with the number of nodes. In our work, we study the more common regime, for example in the work of Penrose \cite{penrose2016connectivity}, where the `connection range' of each node is decreasing in the number of nodes. Equivalently, the size of the domain in which we embed the graph is growing as we add more nodes. We make this idea more precise later. To the best of our knowledge, this regime has not yet been studied in an information-theoretic context.

We also note that the SRGG is an example of a \textit{non-standard source}, in that it cannot be written as a string of stationary, ergodic or Markovian symbols. We must therefore employ the ideas of information spectrum theory to study its information theoretic properties. In this work, we contribute the following. \begin{enumerate}
    \item A limit theorem for the entropy rate of the SRGG,
    \item an Asymptotic Equipartition Property (AEP) for the SRGG,
    \item a complete characterisation of the Slepian-Wolf region for the distributed compression of the SRGG.
\end{enumerate} In Section \ref{sec:prelims} we formally introduce the model, the information spectrum, and the distributed compression setting. Section \ref{sec:results} presents our main results, and Section \ref{sec:proofs} contains proof ideas. The full proofs, and plots of the rate region may be found in the extended version on arXiv \cite{baker2026entropydistributedsourcecoding}. Finally, Section \ref{sec:conclusion} concludes.

\section{Model and Problem Formulation}
\label{sec:prelims}

\subsection{Soft Random Geometric Graphs}

\begin{definition}[Soft Random Geometric Graph]
    Let $Z_1,...,Z_n$ be independent, uniformly distributed points in a compact set $\mathcal{K} \subset \mathbb{R}^d$ with unit volume, a non-empty interior, whose boundary is piecewise smooth, of finite length and has a finite number of corners. The Soft Random Geometric Graph (SRGG) $G_n = (V_n, E_n)$ is the graph with vertex set $V_n = [n] := \{1,2,...,n\}$, and edge set formed by connecting each pair of vertices $i, j \in V_n$ independently with probability $p_n(R_{ij})$, where $R_{ij} := \|Z_i-Z_j\|$ is the Euclidean distance between nodes $i$ and $j$, and $p_n : [0,\infty) \to [0,1]$ is the `connection function'. Let $\graph_n(\mathcal{K})$ be the set of all SRGGs on $n$ nodes embedded in $\mathcal{K}$, then $\graph_n(\mathcal{K})$ is called the ensemble of SRGGs. We will denote this by $\graph_n$ in the remainder of this paper for ease of notation.
\end{definition}
In this work, we will consider specifically the following form of connection function. Let $\{s(n)\}_{n=1}^\infty$ be a positive, decreasing sequence of real numbers. We call $s(n)$ the \textit{sparsity} of the SRGG. Then, we focus on connection functions $p_n(r) = p(r/s(n))$, where $p : [0,\infty) \to [0,1]$. The connection function $p$ is usually non-increasing, with a common example being the Rayleigh fading connection function $p(r/s(n)) = \exp(-(r/s(n))^\eta)$, where $\eta > 0$ controls the decay of connectivity. This can be thought of as decreasing the `connection range' of each node as $n\tti$, that is, the length scale at which two nodes are likely to be connected is decreasing with $n$. We note that this is a different regime to the other commonly considered SRGG, where $p_n(r) = s(n)p(r)$. We will require that $1 \gg s(n) \gg \left(\frac{\log n}{n}\right)^{1/d}$. This requirement has two interpretations. The first is that this is the regime where the SRGG is connected with high probability as $n\tti$ \cite{penrose2016connectivity}. The second is that the entropy of the labelled graph is asymptotically equivalent to the entropy of the unlabelled graph.

\begin{definition}[Entropy of the SRGG]
    The Shannon Entropy of the SRGG is denoted by $H(\graph_n)$, and is given by
    \begin{equation}
        H(\graph_n) := - \sum_{G \in \graph_n} \prob{G_n}\log \prob{G_n}
    \end{equation}
    where the logarithm is to base 2, and we take the convention that $0 \log 0 = 0$. Also, if $Z_n$ is the random vector containing the locations of each node, then the conditional entropy $H(\graph_n|\mathcal{Z}_n) = \mathbb{E}[H(\graph_n | Z_n)]$ is given by
    \begin{equation}
        H(\graph_n|\mathcal{Z}_n) = \binom{n}{2}\int_0^\infty f_{\mathcal{K}}(r)h_2(p_n(r))dr
    \end{equation}
    where $f_{\mathcal{K}}$ is the probability density of pairwise distances in $\mathcal{K}$, and $h_2(p) = -p\log p-(1-p)\log(1-p)$ is the binary entropy function.
\end{definition}

We treat the following as a standing assumption for the remainder of the paper.

\begin{assumption}
    \label{ass:1}
    The connection function $p_n$ is given by $p_n(r) = p(r/s(n))$, where $p$ is a H\"older continuous function, that is, for every $(x,y),(x',y') \in \mathcal{K}^2$,
    \begin{equation}
        |p(x,y) - p(x',y')| \leq L\|(x,y)-(x',y')\|^\alpha
    \end{equation}
    with $L, \alpha > 0$ as constants. We assume that, if $d$ is the dimension of the space $\mathcal{K} \subset \mathbb{R}^d$,
        \begin{equation}
            \label{ass:var_1}
            \int_0^\infty r^{2d}h_2(p(r)) dr < \infty
        \end{equation}
        and
        \begin{equation}
            \label{ass:var_2}
            \int_0^\infty r^{d-1}p(r)(1-p(r))\log^2\left(\frac{p(r)}{1-p(r)}\right)dr <\infty
        \end{equation}
    We also define the following integral
    \begin{equation}
        \label{ass:h_star}
        h^* := \kappa_{d-1}\int_0^\infty r^{d-1}h_2(p(r))dr < \infty
    \end{equation}
    where $\kappa_{d-1}$ is the volume of the unit $d-1$ ball. 
\end{assumption}

\subsection{The Information Spectrum}

As briefly mentioned in the introduction, the SRGG is what is known as a `non-standard source'. In classical Information Theory, sources are considered to consist of a sequence of i.i.d. symbols, or at the very least, to be stationary and ergodic. However, the SRGG has neither of these properties. Instead, we must consider the graph $G_n$ in its entirety. \textit{Information spectrum theory} was developed to deal with general sources such as the SRGG. The theory is developed around \textit{spectral entropy rates}. We first define the \textit{limit superior in probability}, and the \textit{limit inferior in probability}. Let $(X_1,X_2,...)$ be a real-valued sequence of random variables, then
\begin{equation}
    \plimsup_{n\tti} X_n := \inf \left\{\alpha | \lim_{n\tti}\prob{X_n > \alpha} =0\right\}
\end{equation}
\begin{equation}
    \pliminf_{n\tti} X_n := \sup \left\{\beta | \lim_{n\tti}\prob{X_n < \beta} =0\right\}
\end{equation}

If $\pliminf_{n\tti} X_n = \plimsup_{n\tti} X_n = X^*$, we say that $\plim_{n\tti} X_n = X^*$, which is equivalent to $X_n \to X^*$ in probability as $n\tti$ \cite[\S 1.3]{han_book}.

\begin{definition}[Spectral Entropy Rates]
    Let $\mathbf{G} = \{G_n\}_{n=2}^\infty$ denote a general source (see \cite[\S 1.3]{han_book}), where $G_n$ is a $\graph_n$-valued random variable. We define the spectral sup-entropy rate of $\mathbf{G}$ as
    \begin{equation}
        \overline{H}(\mathbf{G}) = \plimsup_{n\tti} \frac{1}{\binom{n}{2}s(n)^d}\log \frac{1}{\prob{G_n}}
    \end{equation}
    Similarly, the spectral inf-entropy rate is defined as
    \begin{equation}
        \underline{H}(\mathbf{G}) = \pliminf_{n\tti} \frac{1}{\binom{n}{2}s(n)^d}\log \frac{1}{\prob{G_n}}
    \end{equation}
    If $\overline{H}(\mathbf{G}) = \underline{H}(\mathbf{G})$, then the limit in probability exists, and the spectral entropy rate
    \begin{equation}
        H(\mathbf{G}) = \plim_{n\tti} \frac{1}{\binom{n}{2}s(n)^d}\log \frac{1}{\mathbb{P}(G_n)}
    \end{equation}
    is well-defined.
\end{definition}

These quantities are similar to the spectral entropy rates used in information spectrum theory for general sources \cite{han_book}. The difference is in the normalisation. Normalising by the blocklength $\binom{n}{2}$, would give $\overline{H}(\mathbf{G})=0$. This is not a very interesting result, since the sparsity induced by $s(n)$ means the graph has a vanishing density, and so the majority of edges are absent in the $n\to\infty$ limit. Our definition may be thought of as an `effective' spectral entropy rate, capturing the leading order behaviour of compression length, rather than simply dividing through by the number of symbols in the source. It is simple to check that these quantities retain their operational meaning after replacing the standard scaling of $1/n$ by $1/(\binom{n}{2}s(n)^d)$. \\

\vspace{-0.2cm}
\section{Main Results}
\label{sec:results}

We first establish the following limit theorem.

\begin{theorem}[Limit Theorem]
    \label{thm:limit}
    Let $\graph_n$ be the ensemble of SRGGs with connection function $p_n(r) = p(r/s(n))$. Then,
    \begin{equation}
        \lim_{n\tti} \frac{H(\graph_n)}{\binom{n}{2}s(n)^d} = h^*
    \end{equation}
\end{theorem}

The next result strengthens the above to an AEP, showing that the spectral sup-entropy rate is equal to the spectral inf-entropy rate.

\begin{theorem}[Asymptotic Equipartition Property]
    \label{thm:aep}
    If $G_n$ is a random SRGG sampled from $\graph_n$,
    \begin{equation}
        \plim_{n\tti} \frac{1}{\binom{n}{2}s(n)^d} \log \frac{1}{\prob{G_n}} = h^*
    \end{equation} 
    In other words,
    \begin{equation}
        \overline{H}(\mathbf{G}) = \underline{H}(\mathbf{G}) = H(\mathbf{G}) = h^*
    \end{equation}
    This implies the existence of sets defined for all $\epsilon > 0$,
    \begin{equation}
        \mathcal{T}_n^\epsilon := \left\{G_n : \left|\log \frac{1}{\prob{G_n}} - H(\graph_n)\right| \leq \binom{n}{2}s(n)^d \epsilon\right\} 
    \end{equation}
    such that
    \begin{enumerate}
        \item $\prob{G_n \in \mathcal{T}_n^\epsilon} \to 1$ as $n\tti$
        \item For each $G_n \in \mathcal{T}_n^\epsilon$, $2^{-\binom{n}{2}(s(n)^dh^* +\epsilon)} \leq \prob{G_n} \leq 2^{-\binom{n}{2}s(n)^d(h^*- \epsilon)}$
        \item For every $\delta > 0$, there exists an $n^*$ so that for all $n > n^*$, we have
        \begin{equation}
            (1-\delta)2^{\binom{n}{2}s(n)^d(h^*-\epsilon)} \leq |\mathcal{T}_n^\epsilon| \leq 2^{\binom{n}{2}s(n)^d(h^*+\epsilon)}
        \end{equation}
    \end{enumerate}
\end{theorem}

We now introduce the distributed compression scheme for the SRGG. For each node $i \in [n]$, let $N_i$ be a $\{0,1\}^{(n-1)}$-valued random vector representing the neighbourhood of $i$ in $G_n$. That is, $N_i = (X_{ij})_{j \neq i}$, where $X_{ij} = 1$ if $i$ and $j$ share an edge, and is $0$ otherwise. In this paper we consider $X_{ij}$ to be the same random variable as $X_{ji}$. The distributed compression problem is formulated as follows. Let $L \in \mathbb{N}$ be a constant. We partition the graph $G_n$ into $L$ `blocks', given by
\begin{equation}
    B^{(n)}_{l} = \left\{N_i : (l-1)n/L < i \leq ln/L\right\}
\end{equation}
for $l = 1,...,L$. Denote by $\mathcal{B}_l^{(n)}$ the set of all realisations of $B^{(n)}_l$ for each $l$. The distributed compression problem asks whether there exists a sequence of encoders $\{(\phi_1^{(n)},...,\phi_L^{(n)})\}_{n=1}^\infty$, where $\phi_l^{(n)} : \mathcal{B}_l^{(n)} \to \mathcal{M}_l^{(n)}$, and a sequence of decoders $\{\psi^{(n)}\}_{n=1}^\infty$ with $\psi^{(n)}: \prod_{l=1}^n\mathcal{M}_l^{(n)} \to \graph_n$ such that
\begin{equation*}
    P_E^{(n)} := \prob{\psi^{(n)}(\phi_1^{(n)}(B^{(n)}_1),...,\phi_L^{(n)}(B^{(n)}_L)) \neq G_n} \to 0
\end{equation*}
as $n\tti$. We say an $L$-tuple of rates $(R_1,...,R_L) \in [0,\infty)^{L}$ is \textit{achievable} if 
\begin{equation*}
    \limsup_{n\tti} \frac{L}{n(n-1)s(n)^d} \log |\mathcal{M}_l^{(n)}| \leq R_l
\end{equation*}
for each $l=1,...,n$, and $P_E^{(n)} \to 0$ as $n\tti$. The above normalisation can be understood as $L/n$ for the size of the block, and $1/((n-1)s(n)^d)$ to normalise the per-neighbourhood entropy (consisting of $n-1$ symbols each). For this case, we have the following result
\begin{theorem}[Distributed Compression Rate Region]
    \label{thm:partial_dsc}
    In the distributed compression scheme as described above, a rate tuple $(R_1,...,R_L)$ is achievable if and only if
    \begin{equation}
        \sum_{l \in \Lambda} R_l \geq \frac{|\Lambda|^2}{L}\frac{h^*}{2} 
    \end{equation}
    for every $\Lambda \subseteq [L]$.
\end{theorem}
\begin{remark}
    In particular, when $\Lambda = [L]$, we require $\frac{1}{L}\sum_{l=1}^L R_l \geq \frac{h^*}{2}$. The factor of $1/2$ comes from the fact that the information about each edge is stored twice overall. To see this, note that the variable $X_{12}$ is in both $N_1$, and $N_2$. 
\end{remark}

\section{Proof Ideas}
\label{sec:proofs}

\subsection{Theorem \ref{thm:limit}}
We first prove that $H(\graph_n) = H(\graph_n|\mathcal{Z}_n) + O(n\log n)$. We follow the idea of \cite[Theorem D.5]{janson2010graphons}, which establishes this result in the dense regime ($s(n) = \Theta(1)$). It follows by discretising the domain $\mathcal{K}$ into $m^d \in \mathbb{N}$ boxes, and conditioning on the box that contains each node. The proof of our result requires $m$ to vary with $n$, which is not the case for the dense regime. Explicitly, if $M$ is the random vector collecting these labels, which has entropy $H(M) \leq nd\log m$,
\begin{equation}
    H(\graph_n|\mathcal{Z}_n) \leq H(\graph_n) \leq H(\graph_n|M) + H(M)
\end{equation}
\begin{align}
    \Rightarrow &H(\graph_n)-H(\graph_n|\mathcal{Z}_n) \leq nd\log m + H(\graph_n|M) - H(\graph_n|\mathcal{Z}_n) \nonumber \\
    &\leq nd\log m + \binom{n}{2}(H(X_{12}|M) - H(X_{12}|\mathcal{Z}_n))
\end{align}
Then, letting $m$ grow with $n$, the aim is to show that the second term is $o(s(n)^d)$. $H(X_{12}|M)$ is the entropy of an edge existing with respect to an `averaged' version of $p$, and so using the H\"older continuity of $p$ it can be shown that\vspace{2cm}
\begin{equation}
    H(X_{ij}|M) - H(X_{ij}|\mathcal{Z}_n) \leq (1+o(1)) C\frac{\log (s(n)m)}{s(n)^{\alpha}m^\alpha}.
\end{equation}
for some $C > 0$. This means that we can choose any $m = m(n)$ such that $\log(m) \ll ns(n)^d$ and $m s(n)^{(d+\alpha)/\alpha}\gg 1$ to get
\begin{equation}
    \frac{H(\graph_n)-H(\graph_n|\mathcal{Z}_n)}{\binom{n}{2}s(n)^d} \leq \frac{nd\log m}{\binom{n}{2}s(n)^d}+(1+o(1)) C\frac{\log (s(n)m)}{s(n)^{d+\alpha}m^\alpha}
\end{equation}
which goes to 0 as $n\tti$. To complete the result, we use Theorem 1 from \cite{baker2026entropy}. Under the assumption \eqref{ass:h_star}, this gives
\begin{equation}
    \frac{H(\graph_n|\mathcal{Z}_n)}{\binom{n}{2}s(n)^d} \overset{n\tti}{\to} h^*.
\end{equation}
Combining with the above argument yields the result. \qed

\subsection{Theorem \ref{thm:aep}}

We begin by writing
\begin{align}
    |-\log \prob{G_n} - H(\graph_n)| &\leq |H(\graph_n) - H(\graph_n|\mathcal{Z}_n)| \nonumber \\ &+ |-\log \prob{G_n} - H(\graph_n | \mathcal{Z}_n)| \nonumber 
\end{align}
By Theorem \ref{thm:limit} the first term is $o(n^2s(n)^d)$, and for the second term, we use Markov's inequality to give
\begin{align}
    \mathbb{P}(|-\log\space &\prob{G_n} - H(\graph_n|\mathcal{Z}_n)| > \binom{n}{2}s(n)^d\epsilon) \nonumber \\&\leq \frac{\mathbb{E}_{G_n}[|-\log \prob{G_n} - H(\graph_n|\mathcal{Z}_n)|]}{\binom{n}{2}s(n)^d\epsilon} \nonumber \\
    &= \frac{\mathbb{E}_{G_n,Z_n}[|-\log \prob{G_n} - H(\graph_n|\mathcal{Z}_n)| | Z_n]}{\binom{n}{2}s(n)^d\epsilon} \nonumber \\ 
    \label{eq:split_bound}
    &\leq \frac{\mathbb{E}_{G_n,Z_n}[|-\log \prob{G_n|Z_n} - H(\graph_n|\mathcal{Z}_n)|]}{\binom{n}{2}s(n)^d\epsilon} \nonumber \\ &+ \frac{\mathbb{E}_{G_n,Z_n}[|\log \prob{G_n|Z_n} - \log \prob{G_n}|]}{\binom{n}{2}s(n)^d\epsilon}
\end{align}

To show that the first term goes to zero, we need the following lemma.
\begin{lemma}[Conditional AEP]
    \label{thm:conditional_aep}
    Let $\{(G_n,Z_n)\}_{n=2}^\infty$ be a sequence of pairs of random sparse SRGGs and the node locations that generated it. Then,
    \begin{equation}
        s(n)^{-d}\binom{n}{2}^{-1}\left(-\log \prob{G_n|Z_n} - H(\graph_n|\mathcal{Z}_n)\right) \to 0
    \end{equation}
    in expectation as $n\tti$.
\end{lemma}

 This is easily proved using the weak law of large numbers. The only condition we need to verify is that:
\begin{align*}
    &\frac{\Var(\log \prob{G_n|Z_n})}{\binom{n}{2}^2s(n)^{2d}} \nonumber \\ &= \frac{\Var(X_{12}\log \prob{X_{12}|Z_n}+ (1-X_{ij})\log(1-\mathbb{P}(X_{12}|Z_n)))}{\binom{n}{2}s(n)^{2d}}
\end{align*}
is finite, which allows us to use Chebyshev's inequality. This is proven by a similar argument to the proof of \cite[Theorem 1]{baker2026entropy} cited above, and uses assumptions \eqref{ass:var_1} and \eqref{ass:var_2}. This condition also allows us to upgrade convergence in probability (given by WLLN) to convergence in expectation via de la Vall\'ee Poussin's criterion for uniform integrability.

The second term goes to 0 via the following argument. Let
\begin{equation}
    r_n = r(G_n, Z_n) = \frac{\prob{G_n|Z_n}}{\prob{G_n}}.
\end{equation}
Then, we have
\begin{equation*}
    \mathbb{E}_{G_n,Z_n}[|\log \prob{G_n|Z_n} - \log \prob{G_n}|] = \mathbb{E}_{G_n,Z_n}[|\log r_n|]
\end{equation*}
\begin{equation}
    = \mathbb{E}_{Z_n}\left[\sum_{G_n \in \graph_n} \prob{G_n|Z_n} |\log r_n|\right]
\end{equation}
\begin{equation}
    = \mathbb{E}_{Z_n}\left[\sum_{G_n \in \graph_n} \prob{G_n} |r_n \log r_n| \right]
\end{equation}
Then, we make the following two observations. First, $|x| = x - 2x\mathds{1}(x < 0)$, and second, for $x \in (0,1)$, $0 < -x\log x < \frac{1}{e}\log e$. Therefore,
\begin{align}
    \mathbb{E}_{G_n,Z_n}[|\log r_n|] \leq \mathbb{E}_{Z_n}\left[\sum_{G_n \in \graph_n} \prob{G_n} r_n \log r_n \right] \nonumber \\ - 2\mathbb{E}_{Z_n}\left[\sum_{G_n \in \graph_n} \prob{G_n} r_n \log r_n \Bigg| r_n < 1 \right] \nonumber \\
    \leq \mathbb{E}_{Z_n}\left[\sum_{G_n \in \graph_n} \prob{G_n} r_n \log r_n \right] \nonumber \\+ \frac{2}{e}\log(e) \mathbb{E}_Z\left[\sum_{G_n \in \graph_n} \prob{G_n} \right] \nonumber \\
    = \mathbb{E}_{G_n,Z_n}\left[\log \frac{\prob{G_n|Z_n}}{\prob{G_n}}\right] + \frac{2}{e}\log e
\end{align}
\begin{equation}
    = H(\graph_n) - H(\graph_n|\mathcal{Z}_n) + \frac{2}{e}\log e = o(n^2s(n)^d)
\end{equation}
This places an $o(1)$ bound on the probability that $|\log(1/\prob{G_n}) - H(\graph_n)| > \binom{n}{2}s(n)^d\epsilon$ for any $\epsilon>0$. Combining this with the fact that $(\binom{n}{2}s(n)^d)^{-1}H(\graph_n) \to h^*$ as $n\tti$, we get $\plim_{n\tti} \frac{1}{\binom{n}{2}s(n)^d}\log \frac{1}{\prob{G_n}} = h^*$. The properties of the typical set are derived using simple standard techniques.

\subsection{Theorem \ref{thm:partial_dsc}}

\begin{lemma}
    \label{thm:nbhd_reduction}
    Let $S_n \subseteq V_n$, and $S_n^c := V_n \setminus S_n$ then
    \begin{equation}
        H(N_{S_n}|N_{S_n^c}) = H(\graph_n[S_n])
    \end{equation}
    where $\graph_n[S_n]$ is the ensemble of subgraphs $G_n[S_n]$ of $G_n$ restricted to node set $S$. This implies that if $\{S_n\}_{n=1}^\infty$ is a sequence of subsets of $V_n$ indexed by $n$, such that $|S_n| \to \infty$ as $n\tti$, then
    \begin{equation}
        \label{eq:nbhd_aep}
        \plim_{n\tti} \frac{1}{\binom{|S_n|}{2}s(n)^d}\log \frac{1}{\prob{N_{S_n}|N_{S_n^c}}} = h^* 
    \end{equation}
\end{lemma}

\begin{proof}
    For the first part, fix $n$, and let $k < n$. We will show the first part for ${S} = \{1,..,k\}$,  which generalises to any set ${S}$ by re-labelling the graph. Writing out $N_{{S}}$ in full, we have, recalling that $X_{ij}$ is identified with $X_{ji}$,
    \begin{equation}
        N_{{S}} = \{X_{ij}\}_{\substack{i \in {S}, 1 \leq j \leq n\\
                                    i \neq j}}, N_{{S}^c} = \{X_{ij}\}_{\substack{i \notin {S}, 1 \leq j \leq n\\
                                    i \neq j}}
    \end{equation}
    Then, the entropy $H(N_{S}|N_{S^c
    })$ is the entropy of $A_{{S}} := N_{S} \setminus (N_{S} \cap N_{{S}^c})$ conditioned on $\widetilde{A}_{S} := N_{{S}^c} \setminus (N_{{S}} \cap N_{{S}^c})$, since the edges in the intersection do not contribute to the entropy. The key observation is that if $X_{ij} \in A_{S}$, then there does not exist an edge $X_{uv}$ in $\widetilde{A}_{S}$ with $u = i$ or $v = j$. Therefore, every edge variable in $A_{S}$ is independent of all those in $\widetilde{A}_{S}$. Then,
    \begin{equation*}
        H(N_S|N_{S^c}) = H(A_{S}|\widetilde{A}_{S}) = H(\{X_{ij}\}_{1\leq i<j \leq k}) = H(\graph_n[S])
    \end{equation*}
    as required. For the second part of the theorem, let $\{S_n\}$ be a sequence of sets with size proportional to $n$. Since each edge is added to the graph independently, the subgraph $G_n[S]$ is marginally distributed as a SRGG on $|S|$ nodes with sparsity $s(n)$ inherited from the parent graph $G_n$. Then, we can apply Theorem \ref{thm:aep} to get the result.
\end{proof}
\begin{proof}[Proof of Theorem \ref{thm:partial_dsc}]
    For $\Lambda \subseteq [L]$, let $B_{\Lambda}^{(n)} = \bigcup_{l\in \Lambda} B_l^{(n)}$, the union of the neighbourhoods of all the nodes in block $l$ for each $l \in \Lambda$. After establishing the above Lemma \ref{thm:nbhd_reduction}, for each $\Lambda \subseteq [L]$, and for all $\epsilon > 0$ we have a typical set $\mathcal{T}_n^\epsilon(\Lambda)$, which we define as
    \begin{align}
        &\mathcal{T}_n^\epsilon(\Lambda) := \Big\{(B_\Lambda^{(n)},B_{\Lambda^c}^{(n)}) :\nonumber \\ &\Big| \frac{1}{{\binom{{|\Lambda|n}/{L}}{2} s\left(n\right)^d}}\log {\mathbb{P}(B_{\Lambda}^{(n)}|B_{\Lambda^c}^{(n)})} 
        -h^* \Big| < \epsilon\Big\}
    \end{align}
    whose existence is guaranteed by the second part of Lemma \ref{thm:nbhd_reduction}.
    With this, we can apply the standard random binning argument, which we sketch here. Let $L$ be a fixed integer, then we design the codebook as follows for an $L$-tuple of rates $(R_1,...,R_L)$. For each $l = 1,...,L$, we generate assignments by independently assigning each $b_l^{(n)} \in \mathcal{B}_l^{(n)}$ a uniformly random element, called its index $\phi_l^{(n)}(b_l^{(n)})$, from the codebook
    \begin{equation*}
        \mathcal{M}_l^{(n)} := \{1,2,3,...,2^{\frac{n}{L}(n-1)s(n)^dR_l}\}.
    \end{equation*}
    Note that more than one realisation $b_l^{(n)}$ may be assigned the same index. We reveal these assignments to both the encoder, and the decoder. To encode a graph, each encoder $l$ sees its block, $b_l^{(n)}$. Then, the independent encodings are sent as the tuple $(\phi_1^{(n)}(b_1^{(n)}),...,\phi_L^{(n)}(b_L^{(n)}))$. The decoder receives this tuple, and searches for their own tuple $(\hat{b}_1^{(n)},...,\hat{b}_L^{(n)})$ (or equivalently a graph $\hat{G}_n$), such that $(\hat{b}_1^{(n)},...,\hat{b}_L^{(n)})$ is unique, typical, and $\phi_l^{(n)}(\hat{b}_l^{(n)}) = \phi_l^{(n)}(b_l^{(n)})$ for every $l$. If they cannot find such a tuple, an error is declared. Therefore, an error occurs if, first, the graph is atypical, or second, there exists an alternative decoding that is typical, and matches the encoding of the target graph. By Theorem \ref{thm:aep} the first error probability goes to 0. The probability of the second error event is bounded above by the sum over such $\Lambda$ of
    \begin{align}
        \prob{E_2(\Lambda)} = \mathbb{P}\Big(\exists &\hat{b}_\Lambda^{(n)}\neq B_\Lambda^{(n)} \space \text{ s.t. } (B_{\Lambda}^{(n)}, B_{\Lambda^c}^{(n)}) \in \mathcal{T}_n^\epsilon(\Lambda) \nonumber \\ &\text{ and } \phi_l^{(n)}(\hat{b}_l^{(n)}) = \phi_l^{(n)}(B_l^{(n)})\space \forall l \in \Lambda\Big) \nonumber 
    \end{align}
    Following the standard argument, this is bounded above by
    \begin{equation}
        \prob{E_2(\Lambda)} \leq 2^{-\frac{n(n-1)}{L}s(n)^d\sum_{l\in \Lambda}R_l}|\mathcal{T}_n^\epsilon(\Lambda)|
    \end{equation}
    \begin{equation}
        \leq 2^{-\frac{n(n-1)}{L}s(n)^d\sum_{l\in \Lambda}R_l + \binom{v}{2}s(n)^d(h^*+\epsilon)}
    \end{equation}
    where $v = |B_{\Lambda}^{(n)}| = |\Lambda|n/L$. From this we can see that a sufficient condition for achievability w.h.p. as $n\tti$ is
    \begin{equation}
        \sum_{l\in \Lambda} R_l \geq \frac{|\Lambda|^2}{L}\frac{h^*}{2}
    \end{equation}
    The proof of the converse is standard and may be found in the extended version. The important point is that
    \begin{align}
        &\lim_{n\tti} P_E^{(n)} \nonumber \\ &\geq \lim_{n\tti} \mathbb{P}\Big(\frac{1}{{\frac{n}{L}}(n-1)s(n)^d} \log \frac{1}{\mathbb{P}(B_{\Lambda}^{(n)}|B_{\Lambda^c}^{(n)})} \geq \sum_{l \in \Lambda} R_l + \gamma\Big) \nonumber \\
        &= \mathds{1}\left(\lim_{n\tti}\frac{2}{|\Lambda|}\frac{(n-1)}{\left(|\Lambda|n/{L}-1\right)}\left(\sum_{l \in S} R_l + \gamma\right) < h^*\right)
    \end{align}
    which follows from \cite[Lemma 7.2.2]{han_book}, the definition of $\plim$ and Lemma \ref{thm:nbhd_reduction}, and gives that any achievable rate tuple must satisfy
    \begin{equation}
        \sum_{l \in \Lambda} R_l \geq \frac{|\Lambda|^2}{L}\frac{h^*}{2}
    \end{equation}
    for the matching converse.
\end{proof}

\section{Conclusion}
\label{sec:conclusion}

In this paper, we have studied the information theoretic properties of the soft random geometric graph in the sparse super-connectivity regime. We gave limit results which characterise the entropy rate of the SRGG, and establish that it is an information-stable source, in that its information content converges to the entropy rate with high probability. Next, we found the Slepian-Wolf rate region for distributed compression of the SRGG where there are a finite number of blocks of nodes. For future research, a tractable distributed compression algorithm would be a practically useful direction. From a theoretical standpoint, it would be interesting to extend our results to the regime where the number of compression units scales with $n$, and compare this region to the finite block size we considered. Extensions to lossy distributed compression may be of interest, including Berger-Tung bounds for this problem for a range of distortion functions.

\section*{Acknowledgment}
 The authors would like to thank Yihan Zhang for his helpful discussion. OB was funded by the EPSRC Centre for Doctoral Training in Computational Statistics and Data Science (COMPASS), Grant Number EP/S023569/1.
 
\bibliographystyle{IEEEtran}
\bibliography{bibliofile}

\end{document}


\maketitle

This supplement contains the full proofs of Theorems 1,2 and 3 in the main document: ``\textbf{{Entropy and Distributed Source Coding of Connected Soft Random Geometric Graphs}}".

\section{Assumption and Statement of Results}
\label{sec:proofs}

We first re-state our standing assumption, and our main results:
\begin{assumption}
    \label{ass:1}
    The connection function $p_n$ is given by $p_n(r) = p(r/s(n))$, where $p$ is a H\"older continuous function, by which we mean the graphon $W(x,y) = p(\|x-y\|)$ satisfies,
    \begin{equation}
        |W(x,y) - W(x',y')| \leq L\|(x,y)-(x',y')\|^\alpha
    \end{equation}
     for every $(x,y),(x',y') \in \mathcal{K}^2$, where $L, \alpha > 0$ are constants. We assume that, if $d$ is the dimension of the space $\mathcal{K} \subset \mathbb{R}^d$,
        \begin{equation}
            \label{ass:var_1}
            \int_0^\infty r^{2d}h_2(p(r)) dr < \infty
        \end{equation}
        and
        \begin{equation}
            \label{ass:var_2}
            \int_0^\infty r^{d-1}p(r)(1-p(r))\log^2\left(\frac{p(r)}{1-p(r)}\right)dr <\infty
        \end{equation}
    We also define the following integral
    \begin{equation}
        \label{var:h_star}
        h^* := \kappa_{d-1}\int_0^\infty r^{d-1}h_2(p(r))dr < \infty
    \end{equation}
    where $\kappa_{d-1}$ is the volume of the unit $d-1$ ball. 
\end{assumption}

\begin{theorem}[Limit Theorem]
    \label{thm:limit}
    Let $\graph_n$ be the ensemble of SRGGs with connection function $p_n(r) = p(r/s(n))$. Then,
    \begin{equation}
        \lim_{n\tti} \frac{H(\graph_n)}{\binom{n}{2}s(n)^d} = h^*
    \end{equation}
\end{theorem}

\begin{theorem}[Asymptotic Equipartition Property]
    \label{thm:aep}
    If $G_n$ is a random SRGG sampled from $\graph_n$,
    \begin{equation}
        \plim_{n\tti} \frac{1}{\binom{n}{2}s(n)^d} \log \frac{1}{\prob{G_n}} = h^*
    \end{equation} 
    In other words,
    \begin{equation}
        \overline{H}(\graph_n) = \underline{H}(\graph_n) = h^*
    \end{equation}
    This implies the existence of sets defined for all $\epsilon > 0$,
    \begin{equation}
        \mathcal{T}_n^\epsilon := \left\{G_n : \left|\log \frac{1}{\prob{G_n}} - H(\graph_n)\right| \leq \binom{n}{2}s(n)^d \epsilon\right\} 
    \end{equation}
    such that
    \begin{enumerate}
        \item $\prob{G_n \in \mathcal{T}_n^\epsilon} \to 1$ as $n\tti$
        \item For each $G_n \in \mathcal{T}_n^\epsilon$, $2^{-\binom{n}{2}(s(n)^dh^* +\epsilon)} \leq \prob{G_n} \leq 2^{-\binom{n}{2}s(n)^d(h^*- \epsilon)}$
        \item For every $\delta > 0$, there exists an $n^*$ so that for all $n > n^*$, we have
        \begin{equation}
            (1-\delta)2^{\binom{n}{2}s(n)^d(h^*-\epsilon)} \leq |\mathcal{T}_n^\epsilon| \leq 2^{\binom{n}{2}s(n)^d(h^*+\epsilon)}
        \end{equation}
    \end{enumerate}
\end{theorem}

\begin{theorem}[Distributed Compression Rate Region]
    \label{thm:partial_dsc}
    In the distributed compression scheme, a rate tuple $(R_1,...,R_L)$ is achievable if and only if
    \begin{equation}
        \sum_{l \in \Lambda} R_l \geq \frac{|\Lambda|^2}{L}\frac{h^*}{2}
    \end{equation}
    for every $\Lambda \subseteq [L]$.
\end{theorem}

\section{Proof of Theorem \ref{thm:limit}}
We begin with a useful bound involving the binary entropy function.
\begin{lemma}
    \label{thm:h2_bound}
    For $x,y \in [0,1]$,
    \begin{equation}
        \label{bin_bound}
        |h_2(x)-h_2(y)| \leq h_2(|x-y|).
    \end{equation}
\end{lemma}

\begin{proof}
    Without loss of generality, we can assume $0 < x < y < 1$ due to symmetry. First, rewrite $h(y)-h(x)$ in terms of the derivative of $h_2$.
    \begin{equation}
        h_2(y)-h_2(x) = \int_0^x h_2'(t)dt + \int_x^{y-x}h_2'(t)dt + \int_{y-x}^yh_2'(t)dt - \int_0^x h_2'(t)dt
    \end{equation}
    \begin{equation}
        = \int_{0}^{y-x}h_2'(t)dt - \int_0^x h_2'(t)dt + \int_{y-x}^y h_2'(t)dt .
    \end{equation}
    With a change of variables, we can rewrite this as
    \begin{equation}
        h_2(y)-h_2(x) = h_2(y-x) - \int_0^x [h_2'(t) - h_2'(t+y-x)]dt
    \end{equation}
    \begin{equation}
        \leq h_2(y-x)
    \end{equation}
    since $h_2$ is concave, and so its derivative is decreasing. Now, pick $a,b \in [0,1]$, and set $x = \min\{a,b,1-a,1-b\}$ and $y= x+|a-b|$. Note that this implies that $x < y$ and $x < 1-y$. If $x = a$, $y = a + b - a = b$, and $x < 1-y$ by the construction of $x$. If $x = 1-a$, $y = 1 - a + a - b = 1 - b$, so again $x < 1-y$. The same works by symmetry if $x = b$ or $1-b$. Then because $x<y$ and $x < 1-y$ we have $h(y) \geq h(x)$,
    \begin{equation}
        |h_2(a)-h_2(b)| = h_2(y)-h_2(x) \leq h_2(y-x)
    \end{equation}
    by above, and then because $y > x$, finally
    \begin{equation}
       |h_2(a)-h_2(b)| \leq h_2(|b-a|)
    \end{equation}
    completing the proof.
\end{proof}

We will first establish the following limit:
\begin{equation}
    \lim_{n\tti} \frac{H(\graph_n) - H(\graph_n | \mathcal{Z})}{\binom{n}{2}s(n)^d} = 0
\end{equation}
then use this fact and the triangle inequality to prove convergence to a constant. We first note that
\begin{equation}
    H(\graph_n) \geq H(\graph_n|\mathcal{Z}_n)
\end{equation}
since conditioning reduces entropy. Let $\ell > 0$ be the side length of the largest $d$-cube completely containing $\mathcal{K}$, let $m \in \mathbb{N}$, and cover $\mathcal{K}$ with disjoint, equally sized $d$-cubes of side length $\ell/m$, of which we will need $\eta \leq m^d$. Then, for randomly distributed points $Z_1,...,Z_n \in \mathcal{K}$, denote by $M = (M_1,...,M_n)$ the random vector of boxes that each node falls into. Then, it is clear that
\begin{equation}
    H(\graph_n) \leq H(\graph_n|M) + H(M) \leq H(\graph_n | M) + nd\log m
\end{equation}
where the second inequality follows because the nodes are uniformly distributed on $\mathcal{K}$. This means that the entropy of the distribution of the box vector is upper bounded by the joint entropy of the $n$ independent, uniformly distribution on $m^d$ variables, which is $nd\log m$. So, the joint entropy of $n$ uniformly distributed random variables which take on any of $m^d$ values is simply $nd\log m$. Next, we can use the independence bound to give
\begin{equation}
    H(\graph_n|M) \leq \binom{n}{2} H(X_{12} | M_1, M_2).
\end{equation}
Now, define the following graphon,
\begin{equation}
    W_n^{(m)}(x,y) := \begin{cases} p_n(\|x-y\|) & \text{if } (x,y) \in  \mathcal{K}^2 \\ 0 & \text{otherwise} \end{cases}
\end{equation}
and define the following 
\begin{equation}
    w_n^{(m)}(k,l) := \mathbb{E}[W_n^{(m)}(Z_1, Z_2) | Z_1 \in M_k \cap \mathcal{K}, Z_2 \in M_l \cap \mathcal{K}]
\end{equation}
\begin{equation}
    \label{eq:nasty_denominator}
    = \frac{\int\int_{M_1 \times M_2} W_n^{(m)}(x,y)dxdy}{\lambda(M_k \cap \mathcal{K})\lambda(M_l \cap \mathcal{K})},
\end{equation}
where $\lambda(\cdot)$ denotes the Lebesgue measure, and for $x \in M_k$, $y \in M_l$ define
\begin{equation}
    \overline{W}_n^{(m)}(x,y) := w_n^{(m)}(k,l).
\end{equation}
This gives
\begin{equation}
    \overline{W}_n^{(m)}(x,y) = \mathbb{E}[W_n^{(m)}(Z_1, Z_2) | M_1\cap \mathcal{K}, M_2 \cap \mathcal{K}],
\end{equation}
and so
\begin{equation}
    \prob{X_{12}=1 | M_1 = m_k, M_2 = m_l} = w_n^{(m)}(k,l)
\end{equation}
and
\begin{equation}
    \mathbb{E}[H(X_{12} | M_1, M_2)] = \sum_{k,l=1}^{\eta} \prob{M_1 = m_k, M_2 = m_l} H(X_{12} | M_1 = m_k, M_2=m_l)
\end{equation}
\begin{equation}
    = \sum_{k,l=1}^{\eta} \int\int_{M_k \times M_l}dxdy \text{ } h_2(w_n^{(m)}(k,l))
\end{equation}
\begin{equation}
    = \sum_{k,l=1}^{\eta} \int\int_{M_k \times M_l} h_2(\overline{W}_n^{(m)}(x,y))dxdy= \int\int_{\mathcal{K}^2} h_2(\overline{W}_n^{(m)}(x,y))dxdy.
\end{equation}
since $\mathcal{K}$ has unit volume. Then, we get
\begin{equation}
    H(\graph_n|\mathcal{Z}_n)\leq H(\graph_n) \leq nd\log m + H(\graph_n | M)
\end{equation}
\begin{equation}
    \iff \int\int_{\mathcal{K}^2} h_2(W_n(x,y)) dxdy \leq \frac{H(\graph_n)}{\binom{n}{2}} \leq \frac{nd\log m}{\binom{n}{2}} + \int\int_{\mathcal{K}^2} h_2(\overline{W}_n^{(m)}(x,y))dxdy 
\end{equation}
\begin{align}
    \label{eq:entropy_diff_bound}
    \iff \frac{H(\graph_n) - H(\graph_n|\mathcal{Z}_n)}{\binom{n}{2}} \leq &\space \frac{nd\log m}{\binom{n}{2}} \nonumber \\ &+ \int\int_{\mathcal{K}^2}h_2(\overline{W}_n^{(m)}(x,y))dxdy  - \int\int_{\mathcal{K}^2} h_2(W_n(x,y)) dxdy.
\end{align}
So the quantity to bound is the difference of the two integrals in \eqref{eq:entropy_diff_bound}:
\begin{equation}
    0 \leq \int\int_{\mathcal{K}^2}h_2(\overline{W}_n^{(m)}(x,y))dxdy  - \int\int_{\mathcal{K}^2} h_2(W_n(x,y)) dxdy
\end{equation}
\begin{equation}
    \leq \int \int_{\mathcal{K}^2} \left|h_2(\overline{W}_n^{(m)}(x,y)) - h_2(W_n(x,y))\right|dxdy 
\end{equation}
\begin{equation}
\label{eq:int_to_bound}
    \leq \int \int_{\mathcal{K}^2} h_2(|\overline{W}_n^{(m)}(x,y) - W_n(x,y)|)dxdy. 
\end{equation}
where we have used Lemma \ref{thm:h2_bound}. Next, observe that
\begin{equation}
    \mathcal{K} \subseteq \bar{\mathcal{K}}_m \cup \delta \mathcal{K}_m
\end{equation}
where $\bar{K}_m$ is the set of boxes completely contained within $\mathcal{K}$, and $\delta \mathcal{K}_m$ is the set of boxes that overlap the boundary of $\mathcal{K}$. Then,
\begin{align}
    \int \int_{\mathcal{K}^2} h_2(|\overline{W}_n^{(m)}(x,y) - W_n(x,y)|)dxdy \leq \int \int_{\bar{\mathcal{K}}_m^2} h_2(|\overline{W}_n^{(m)}(x,y) - W_n(x,y)|)dxdy \nonumber \\
    + 2\int \int_{\bar{\mathcal{K}}_m \times \delta \mathcal{K}_m} h_2(|\overline{W}_n^{(m)}(x,y) - W_n(x,y)|)dxdy + \int \int_{\delta \mathcal{K}_m^2} h_2(|\overline{W}_n^{(m)}(x,y) - W_n(x,y)|)dxdy.
\end{align}
Using the H\"older property of the graphon $W_n$, we have for $(x,y) \in \mathcal{K}^2$
\begin{equation}
    |W_n(x,y) - W_n^{(m)}(x,y)| \leq \sup_{(x,y),(x',y') \in \mathcal{K}^2} |W_n(x,y) -W_n(x',y')|
\end{equation}
\begin{equation}
    \leq \sup_{(x,y), (x',y') \in \mathcal{K}^2} L\left(\frac{\|(x,y)-(x',y')\|}{s(n)}\right)^\alpha
    \leq L\left(\frac{\ell\sqrt{2d}}{s(n)m}\right)^\alpha.
\end{equation}
Next, with $e \approx 2.718...$ as Euler's number, write
\begin{equation}
    h_2(p) = -p\log p + \frac{(1-p)}{\log_e(2)} \sum_{k=1}^\infty \frac{p^k}{k}
\end{equation}
\begin{equation}
    \leq -p\log p + \frac{p}{\log_e(2)} - \sum_{k=2}^\infty \left(\frac{1}{k-1}-\frac{1}{k}\right) p^k \leq \frac{p}{\log_e(2)} - p\log p = p\log(e/p)
\end{equation}
 and note that $p\log(p/e)$ is an increasing function of $p$ on $[0,1]$. Now, we will need $m$ to be an increasing function of $n$ such that $ms(n) \to\infty$, so, choose $n$ large enough that $s(n)m \geq L^{\frac{1}{\alpha}}\ell{\sqrt{2d}}$ (so that the argument of $p$ in $p\log(p/e)$ is between 0 and 1). Then we have
\begin{align}
    \sum_{(m_1,m_2) \in \bar{\mathcal{K}}_m^2}\int \int_{m_1 \times m_2} h_2(|\overline{W}_n^{(m)}(x,y) - W_n(x,y)|)dxdy  \nonumber \\
    \leq \sum_{(m_1,m_2) \in \bar{\mathcal{K}}_m^2}\int \int_{m_1 \times m_2} L\left(\frac{\ell\sqrt{2d}}{s(n)m}\right)^\alpha \log \left(\frac{e}{L}\left(\frac{s(n)m}{\ell\sqrt{2d}}\right)^\alpha\right)  dxdy  \nonumber \\
\end{align}
\begin{equation}
    =|\bar{\mathcal{K}}_m|^2\frac{ \ell^{2d}}{m^{2d}} \left(L\left(\frac{\ell\sqrt{2d}}{s(n)m}\right)^\alpha \log \left(\frac{e}{L}\left(\frac{s(n)m}{\ell\sqrt{2d}}\right)^\alpha\right)\right)
\end{equation}
\begin{equation}
    \leq \ell^{2d+\alpha}L\left(\frac{\sqrt{2d}}{s(n)m}\right)^\alpha \log \left(\frac{e}{L}\left(\frac{s(n)m}{\ell\sqrt{2d}}\right)^\alpha\right) 
\end{equation}
where the last inequality follows because $|\bar{\mathcal{K}}_m|^2 \leq m^{2d}$. Next, we must bound the sum over $\bar{\mathcal{K}}_m \times \delta \mathcal{K}_m$ and $\delta\mathcal{K}_m^{2}$. Let $|\delta\mathcal{K}|$ be the surface area of $\mathcal{K}$, which we have assumed to be finite. Then, since the side length of each box is $(\ell/m)^{d-1}$, we require at most $O(|\delta\mathcal{K}|(m/\ell)^{d-1})$ boxes to cover the boundary of $\mathcal{K}$, and so $|\delta\mathcal{K}_m \times \bar{\mathcal{K}}_m| = O(m^{2d-1})$, and $|\delta\mathcal{K}_m^2| = O(m^{2d-2})$, so 
\begin{equation}
    \sum_{(m_1,m_2) \in \delta\mathcal{K}_m \times \bar{\mathcal{K}}_m} \int\int_{m_1 \times m_2} h_2(|\overline{W}_n^{(m)}(x,y) - W_n(x,y)|)dxdy  \leq m^{-2d}O(m^{2d-1}) = O(m^{-1}),
\end{equation}
and similarly, the sum over $\delta\mathcal{K}_m^2$ is $O(m^{-2})$. So, combining all of the bounds gives that
\begin{align}
    \frac{H(\graph_n)-H(\graph_n|\mathcal{Z}_n)}{\binom{n}{2}s(n)^d} &\leq \frac{nd\log m}{\binom{n}{2}s(n)^d}  \nonumber \\ &+   \frac{\ell^{2d+\alpha}L\left(\sqrt{2d}\right)^\alpha}{s(n)^{d+\alpha}m^\alpha} \log \left(\frac{e}{L}\left(\frac{s(n)m}{\ell\sqrt{2d}}\right)^\alpha\right)  +O\left(\frac{1}{ms(n)^d}\right).
\end{align}
Therefore, we must choose $m = m(n)$ so that
\begin{equation}
    \label{eq:rate1}
    \frac{\log(s(n)m)}{s(n)^{d+\alpha}m^\alpha} \to 0, 
\end{equation}
\begin{equation}
\label{eq:rate2}
    ms(n)^d \to \infty
\end{equation}
and 
\begin{equation}
\label{eq:rate3}
    \frac{n \log m}{\binom{n}{2}s(n)^d} \to 0
\end{equation}
This gives
\begin{equation}
    \label{eq:to_zero}
    \frac{1}{\binom{n}{2}s(n)^d}|H(\graph_n)-H(\graph_n|\mathcal{Z}_n)| \to 0
\end{equation}
as $n\tti$. We note that \eqref{eq:rate1} - \eqref{eq:rate3} are possible exactly in the regime where $ns(n)^d/\log(n) \to \infty$, but impossible if $ns(n)^d/\log(n) \nrightarrow \infty$, since \eqref{eq:rate3} is then not satisfied. To finish the proof, we must show that $H(\graph_n|\mathcal{Z}_n)/(\binom{n}{2}s(n)^d)$ exists and tends to a non-zero limit. We first note that
\begin{equation}
    \frac{H(\graph_n |\mathcal{Z}_n)}{\binom{n}{2}s(n)^d} = \int\int_{\mathcal{K}^2} h_2(p_n(\|x-y\|))dxdy
\end{equation}
depends on $n$ only through $s(n)$. Therefore, we can use the following result from \cite{baker2026entropy}, which deals with the dense regime in the limit of small connection range, and which we have adapted slightly for our setting. Specifically, \cite[Theorem 1]{baker2026entropy} deals with one explicit connection function, but discusses the general methodology in the subsequent discussion. They note that for general connection functions that satisfy $h^* < \infty$, the limit is given as stated below.
\begin{lemma}[Theorem 1 from \cite{baker2026entropy}]
    Let $\mathcal{K}$ have a non-empty interior, and a boundary that is piecewise smooth, is of finite length, and has a finite number of corners. Then, if $\graph_n$ is the ensemble of SRGGs with a connection function $p(r/r_0)$ (independent of $n$), such that $h^* < \infty$,
    \begin{equation}
        \lim_{r_0 \to 0}\frac{1}{r_0^d} \int\int_{\mathcal{K}^2} h_2(p(\|x-y\|/r_0)) = h^*.
    \end{equation}
\end{lemma}
To apply this to our setting, note that
\begin{equation}
    \frac{H(\graph_n|\mathcal{Z}_n)}{\binom{n}{2}} = H(X_{12} | \mathcal{Z}_n) = \int\int_{\mathcal{K}^2} h_2(p(\|x-y\|/s(n)))dxdy
\end{equation}
This integral depends on $n$ only through $s(n)$. Thus, we can write
\begin{equation}
    \lim_{n\tti} \frac{H(\graph_n)}{\binom{n}{2}s(n)^d} =\lim_{n\tti}  \frac{H(\graph_n) - H(\graph_n|\mathcal{Z}_n)}{\binom{n}{2}s(n)^d} + \lim_{n\tti} \frac{H(\graph_n|\mathcal{Z}_n)}{\binom{n}{2}s(n)^d}
\end{equation}
\begin{equation}
    = 0 + \lim_{r_0 \to 0} \frac{1}{r_0^d} \int\int_{\mathcal{K}^2} h_2(p(\|x-y\|/r_0))dxdy = h^*
\end{equation}
where we have replaced $s(n)$ by $r_0$, and used \eqref{eq:to_zero}. \qed

 \section{Proof of Theorem \ref{thm:aep}}
\label{pf:aep}

\begin{lemma}[Conditional AEP]
    \label{thm:conditional_aep}
    Let $\{(G_n,Z_n)\}_{n=2}^\infty$ be a sequence of pairs of random sparse SRGGs and the node locations that generated it. Then,
    \begin{equation}
        s(n)^{-d}\binom{n}{2}^{-1}\left(-\log \prob{G_n|Z_n} - H(\graph_n|\mathcal{Z}_n)\right) \to 0
    \end{equation}
    almost surely as $n\tti$.
\end{lemma}

\begin{proof}
    For convergence in probability, all we need is
    \begin{equation}
        \label{eq:unconditional_target}
        \mathbb{P}_{G_n,Z_n}\left(\frac{1}{\binom{n}{2}s(n)^d}\left|-\log \mathbb{P}(G_n|Z_n) - H(\graph_n|\mathcal{Z}_n)\right| > \epsilon\right) \to 0
    \end{equation}
    as $n\tti$ for every $\epsilon > 0$. This will follow from a slightly modified law of large numbers. First, we note that, conditioned on $Z_n$ each edge is independent, and so we have
    \begin{equation}
        -\log \prob{G_n|Z_n} = -\sum_{i<j} (X_{ij}\log p_n(\|Z_i-Z_j\|) + (1-X_{ij})\log(1-p_n(\|Z_i-Z_j\|)))
    \end{equation}
    \begin{equation}
        =: \sum_{i<j} Y_{ij}
    \end{equation}
    and
    \begin{equation}
        H(\graph_n|\mathcal{Z}_n) = \binom{n}{2}\mathbb{E}_{Z_n}[h_2(p_n(\|Z_i-Z_j\|))] = \sum_{i<j} \mathbb{E}[Y_{ij}]
    \end{equation}
    Therefore rewriting \eqref{eq:unconditional_target}, our target is now
    \begin{equation}
        \frac{1}{\binom{n}{2}s(n)^d}\left|\sum_{i<j} Y_{ij} - \mathbb{E}[Y_{ij}]\right| \to 0
    \end{equation}
    in probability as $n\tti$. This follows immediately from the Kolmogorov strong law of large numbers, if the variance of $Y_{ij}$ is finite. This can be seen by writing
    \begin{equation}
        \label{eq:cheb_bound}
        \mathbb{P}_{G_n,Z_n}\left(\frac{1}{\binom{n}{2}s(n)^d}\left|\sum_{i<j} Y_{ij} - \mathbb{E}[Y_{ij}]\right| > \epsilon\right)
        \leq \frac{\binom{n}{2}\Var(Y_{ij})}{\binom{n}{2}^2s(n)^{2d}\epsilon^2}
    \end{equation}
    where we have used Chebyshev's inequality, and the independence of each $Y_{ij}$. We must now verify the finiteness of $\Var(Y_{ij})$. Using the law of total variance,
    \begin{equation}
        \Var(Y_{ij}) = \Var(\mathbb{E}[Y_{ij}|Z]) + \mathbb{E}[\Var(Y_{ij}|Z)]
    \end{equation}
    \begin{align}
        \label{eq:variance_expanded}
        = \left[\int\int _{\Omega^2} h_2(p_n(\|Z_i-Z_j\|))dZ_idZ_j \right]^2 - \int\int_{\Omega^2} (h_2(p_n(\|Z_i-Z_j\|))^2dZ_idZ_j \nonumber \\ 
    + \int\int_{\Omega^2}p_n(\|Z_i-Z_j\|)(1-p_n(\|Z_i-Z_j\|))\log^2\left(\frac{p_n(\|Z_i-Z_j\|)}{1-p_n(\|Z_i-Z_j\|)}\right)dZ_idZ_j
    \end{align}
    We will consider the asymptotics of each term in turn. The simplest is the first term. We can use Theorem \ref{thm:limit} to see that
    \begin{equation}
        \int\int_{\Omega^2} h_2(p_n(\|Z_i-Z_j\|)) dZ_idZ_j = \Theta(s(n)^d)
    \end{equation}
    to immediately give
    \begin{equation}
        \left[\int\int_{\Omega^2} h_2(p_n(\|Z_i-Z_j\|)) dZ_idZ_j\right]^2 = \Theta(s(n)^{2d})
    \end{equation}
    Now, note that the second term of \eqref{eq:variance_expanded} is negative and therefore trivially bounded above by 0. After this, we have to deal with
    \begin{equation}
        \int\int_{\Omega^2} p_n(\|Z_i-Z_j\|)(1-p_n(\|Z_i-Z_j\|))\log^2\left(\frac{p_n(\|Z_i-Z_j\|)}{1-p_n(\|Z_i-Z_j\|)}\right) dZ_idZ_j
    \end{equation}
    \begin{equation}
        = \int_0^D f_{\mathcal{K}}(r)I(r/s(n))dr
    \end{equation}
    where for brevity we have defined $I(r/s(n)) := p(r/s(n))(1-p(r/s(n)))\log^2\left(\frac{p(r/s(n))}{1-p(r/s(n))}\right)$. Then
    \begin{equation}
        \int_0^D f_{\mathcal{K}}(r)I(r/s(n))dr = \int_0^{\sqrt{s(n)}} f_{\mathcal{K}}(r)I(r/s(n))dr + \int_{\sqrt{s(n)}}^D f_{\mathcal{K}}(r)I(r/s(n))dr
    \end{equation}
    To control the second term, 
    \begin{equation}
        \int_{\sqrt{s(n)}}^D f_{\mathcal{K}}(r)I(r/s(n))dr \leq \sup_{\sqrt{s(n)} < r < D} |I(r/s(n))| \prob{R > \sqrt{s(n)}}
    \end{equation}
    \begin{equation}
        \leq \sup_{\sqrt{s(n)}<r<D} \left[p(r/s(n))(1-p(r/s(n)))\log^2\left(\frac{p(r/s(n))}{1-p(r/s(n))}\right)\right]
    \end{equation}
    which tends to 0 as $n\tti$. We need the stronger condition that this is $O(s(n)^d)$ as $n\tti$ however. So, since $p$ is a decreasing function, note that for $n$ large enough that $p(1/\sqrt{s(n)}) < \frac{1}{e}\log e$ (where the function $p \mapsto p(1-p)\log^2(p/(1-p))$ has its maximum), we have
    \begin{equation}
        \sup_{t > \frac{1}{\sqrt{s(n)}}}\left[p(t)(1-p(t))\log^2\left(\frac{p(t)}{1-p(t)}\right)\right] = p(1/\sqrt{s(n)})(1-p(1/\sqrt{s(n)}))\log^2 \left(\frac{p(1/\sqrt{s(n)})}{1-p(1/\sqrt{s(n)})}\right)
    \end{equation}
    and so we require $I(1/\sqrt{s(n)}) \ll s(n)^d$, or equivalently $r^{\frac{d-1}{2}}I(r) \ll 1$ as $r\tti$. This is implied by the finiteness of the integral $\int_0^\infty t^{d-1}I(t)dt$. Therefore we have
    \begin{equation}
        \int_{\frac{1}{\sqrt{s(n)}}}^\infty f_{\mathcal{K}}(r)I(r/s(n))dr = o(s(n)^d)
    \end{equation}
    as $n\tti$. For the first term, note
    \begin{equation}
        \int_0^{\sqrt{s(n)}} f_{\mathcal{K}}(r)I(r/s(n))dr = s(n) \int_0^{\frac{1}{\sqrt{s(n)}}} f_{\mathcal{K}}(ts(n)) I(t)dt
    \end{equation}
    \begin{equation}
        = s(n)\int_0^{\frac{1}{\sqrt{s(n)}}} \kappa_{d-1}I(t)s(n)^{d-1}t^{d-1}(1+o(1))dt
    \end{equation}
    \begin{equation}
        = \kappa_{d-1}s(n)^d \int_0^{\frac{1}{\sqrt{s(n)}}} t^{d-1}I(t)dt + o(s(n)^d) \leq \kappa_{d-1}s(n)^d \int_0^\infty t^{d-1}I(t)dt + o(s(n)^d),
    \end{equation}
    which is $\Theta(s(n)^d)$ if $\int_0^\infty t^{d-1}I(t)dt < \infty$ (Assumption \eqref{ass:var_2}). Combining the above gives
    \begin{equation}
        \Var(Y_{ij}) = \Theta(s(n)^d),
    \end{equation}
    and so the Chebyshev bound \eqref{eq:cheb_bound} satisfies
    \begin{equation}
        \frac{\binom{n}{2}\Var(Y_{ij})}{\binom{n}{2}^2s(n)^{2d}\epsilon^2} = \mathcal{O}\left(\left(\binom{n}{2}s(n)^d\right)^{-1}\right)
    \end{equation}
    which goes to $0$ as $n\tti$ by assumption that $ns(n)^d \tti$. Therefore, the probability in \eqref{eq:unconditional_target} goes to 0 and we are done.
\end{proof}

We are now ready to prove the result of Theorem \ref{thm:aep}, which we note is equivalent to 
\begin{equation}
    s(n)^{-d}\binom{n}{2}^{-1}\left(-\log \prob{G_n} - \binom{n}{2}s(n)^dh^*\right) \to 0
\end{equation}
in probability as $n\tti$.

\begin{proof}[Proof of Theorem \ref{thm:aep}]
    We use the triangle inequality to give
    \begin{equation}
        \label{eq:aep_target_to_bound}
        |-\log \prob{G_n} - H(\graph_n)| \leq |H(\graph_n) - H(\graph_n|\mathcal{Z}_n)| + |-\log \prob{G_n} - H(\graph_n | \mathcal{Z}_n)| 
    \end{equation}
    Here, we know the first term is $o(n^2s(n)^d)$ by Theorem \ref{thm:limit}. For the second term, recall that for convergence in probability we require
    \begin{equation}
        \prob{|-\log \prob{G_n} - H(\graph_n|\mathcal{Z}_n)| > \binom{n}{2}s(n)^d\epsilon} \to 0
    \end{equation}
    for all $\epsilon >0$. Indeed by Markov's inequality, followed by the tower law, and triangle inequality,
    \begin{align}
    \prob{|-\log \prob{G_n} - H(\graph_n|\mathcal{Z}_n)| > \binom{n}{2}s(n)^d\epsilon} \leq \frac{\mathbb{E}_{G_n}[|-\log \prob{G_n} - H(\graph_n|\mathcal{Z}_n)|]}{\binom{n}{2}s(n)^d\epsilon} \nonumber \\
    = \frac{\mathbb{E}_{G_n,Z_n}[|-\log \prob{G_n} - H(\graph_n|\mathcal{Z}_n)| | Z_n]}{\binom{n}{2}s(n)^d\epsilon} \nonumber \\ 
    \label{eq:split_bound}
    \leq \frac{\mathbb{E}_{G_n,Z_n}[|-\log \prob{G_n|Z_n} - H(\graph_n|\mathcal{Z}_n)|]}{\binom{n}{2}s(n)^d\epsilon} + \frac{\mathbb{E}_{G_n,Z_n}[|\log \prob{G_n|Z_n} - \log \prob{G_n}|]}{\binom{n}{2}s(n)^d\epsilon}
    \end{align}
    We will show the convergence of each of these terms separately. For the first term, we can use the fact that convergence in probability can be upgraded to convergence in expectation if the family of random variables is uniformly integrable. De la Vall\'ee Poussin's criterion gives that if
\begin{equation}
    \sup_{n\geq 2} \frac{\mathbb{E}[(-\log \prob{G_n|Z_n} - H(\graph_n|\mathcal{Z}_n))^2]}{\binom{n}{2}^2s(n)^{2d}} < \infty,
\end{equation}
and $\plim_{n\tti} (\binom{n}{2}s(n)^d)^{-1}(-\log \prob{G_n|Z_n} - H(\graph_n|\mathcal{Z}_n))=0$ then 
\begin{equation}
    \frac{\mathbb{E}[|-\log \prob{G_n|Z_n) - H(\graph_n|\mathcal{Z}_n)|}]}{\binom{n}{2}^2s(n)^{2d}} \to 0.
\end{equation}
But, this condition is just the finite variance condition we proved in the proof of Lemma \ref{thm:conditional_aep}, combined with Lemma \ref{thm:conditional_aep} itself. Therefore it remains to bound the second term of \eqref{eq:split_bound}. To do this, define
\begin{equation}
    r_n = r(G_n, Z_n) = \frac{\prob{G_n|Z_n}}{\prob{G_n}}.
\end{equation}
Then, we have
\begin{equation}
    \mathbb{E}_{G_n,Z_n}[|\log \prob{G_n|Z_n} - \log \prob{G_n}|] = \mathbb{E}_{G_n,Z_n}[|\log r_n|]
\end{equation}
\begin{equation}
    = \mathbb{E}_{Z_n}\left[\sum_{G_n \in \graph_n} \prob{G_n|Z_n} |\log r_n|\right]
\end{equation}
\begin{equation}
    = \mathbb{E}_{Z_n}\left[\sum_{G_n \in \graph_n} \prob{G_n} |r_n \log r_n| \right]
\end{equation}
Then, we make the following two observations. First,
\begin{equation}
    |x| = x - 2x\mathds{1}(x < 0)
\end{equation}
and second, for $x \in (0,1)$
\begin{equation}
    0 < -x\log x < \frac{1}{e}\log e
\end{equation}
And therefore,
\begin{equation}
    \mathbb{E}_{G_n,Z_n}[|\log r_n|] \leq \mathbb{E}_{Z_n}\left[\sum_{G_n \in \graph_n} \prob{G_n} r_n \log r_n \right] - 2\mathbb{E}_{Z_n}\left[\sum_{G_n \in \graph_n} \prob{G_n} r_n \log r_n \Bigg| r_n < 1 \right]
\end{equation}
\begin{equation}
    \leq \mathbb{E}_{Z_n}\left[\sum_{G_n \in \graph_n} \prob{G_n} r_n \log r_n \right] + \frac{2}{e}\log(e) \mathbb{E}_Z\left[\sum_{G_n \in \graph_n} \prob{G_n} \right]
\end{equation}
\begin{equation}
    = \mathbb{E}_{G_n,Z_n}\left[\log \frac{\prob{G_n|Z_n}}{\prob{G_n}}\right] + \frac{2}{e}\log e
\end{equation}
\begin{equation}
    = H(\graph_n) - H(\graph_n|\mathcal{Z}_n) + \frac{2}{e}\log e = o(n^2s(n)^d)
\end{equation}
Then using this fact, we have that
\begin{equation*}
    \prob{|-\log \prob{G_n} - H(\graph_n|\mathcal{Z}_n)| > \binom{n}{2}s(n)^d\epsilon} \leq \frac{o(n^2s(n)^d)}{\binom{n}{2}s(n)^d\epsilon} \to 0
\end{equation*}
as $n\tti$. Therefore, combining this with the bound in \eqref{eq:aep_target_to_bound}, we have
\begin{equation}
    \prob{|-\log \prob{G_n} - H(\graph_n)| > \binom{n}{2}s(n)^d\epsilon} \to 0
\end{equation}
To complete the result, we write
\begin{align}
    &\plim_{n\tti} \frac{1}{\binom{n}{2}s(n)^d}\log\frac{1}{\prob{G_n}} \nonumber \\ &= \plim_{n\tti} \frac{1}{\binom{n}{2}s(n)^d} \left(\log \frac{1}{\prob{G_n}} - H(\graph_n)\right) + \lim_{n\tti} \frac{1}{\binom{n}{2}s(n)^d}H(\graph_n) \nonumber \\&= 0+h^*=h^*
\end{align}
and we are done. For the second part of the Theorem, properties 1) and 2) follow immediately by rearranging the condition inside the probability in the statement of the Theorem. Property 3) follows by the following observations:
\begin{equation}
    1 \geq \sum_{G_n \in \mathcal{T}_n^\epsilon} \prob{G_n} \geq \sum_{G_n \in \mathcal{T}_n^\epsilon} 2^{-\binom{n}{2}s(n)^d(h^*+\epsilon)} = |\mathcal{T}_n^\epsilon|2^{-\binom{n}{2}s(n)^d(h^*+\epsilon)}
\end{equation}
\begin{equation}
    \Rightarrow |\mathcal{T}_n^\epsilon| \leq 2^{\binom{n}{2}s(n)^d(h^*+\epsilon)}
\end{equation}
and, for $n$ sufficiently large, $\mathbb{P}(G_n \in \mathcal{T}_n^\epsilon) > 1-\delta$, so
\begin{equation}
    1-\delta \leq \sum_{G_n \in \mathcal{T}_n^\epsilon} 2^{-\binom{n}{2}s(n)^d(h^*-\epsilon)} = |\mathcal{T}_n^\epsilon|2^{-\binom{n}{2}s(n)^d(h^*-\epsilon)} 
\end{equation}
    which yields the lower bound.
\end{proof}

\section{Proof of Theorem \ref{thm:partial_dsc}}

After establishing the AEP, we can now turn to distributed coding. The proof of Theorem \ref{thm:partial_dsc} relies on the following observation.

\begin{lemma}
    \label{thm:nbhd_reduction}
    Let $S \subseteq V_n$, then
    \begin{equation}
        H(N_{S}|N_{S^c}) = H(\graph_n[S])
    \end{equation}
    where $\graph_n[S]$ is the ensemble of subgraphs of $G_n[S]$ of $G_n$ restricted to node set $S$.
\end{lemma}
\begin{proof}
    Without loss of generality, we may assume that $S = \{1,...,|S|\}$. Then, writing out the neighbourhood entropy in full, we have
    \begin{equation}
        H(N_{S}|N_{S^c}) = H(N_1,...,N_{|S|}|N_{|S|+1},...,N_n)
    \end{equation}
    The collection $(N_1,...,N_{|S|})$ contains all edge variables in the following set
    \begin{equation}
        \Gamma(S) := \{X_{ij} | i \leq |S|, j \neq i\}
    \end{equation}
    and similarly, $(N_{|S|+1},...,N_n)$ contains edge variables in
    \begin{equation}
        \Gamma(S^c) := \{X_{ij} | i > |S|, j \neq i\}
    \end{equation}
    So $H(N_{S}|N_{S^c}) = H(\Gamma(S)|\Gamma(S^c))$. Therefore, if any variable $X_{ij}$ is in both $\Gamma(S) \cap \Gamma(S^c)$, its information is already accounted for, and does not contributed to the entropy. So, 
    \begin{equation}
        \Gamma(S) \cap \Gamma(S^c) = \{X_{ij} | i \leq |S| \text{ and } j > |S|, \text{ or } i>|S| \text{ and } j \leq |S|\}
    \end{equation}
    which means that
    \begin{equation}
        H(\Gamma(S)|\Gamma(S^c)) = H(\{X_{ij} | i,j \leq |S|, i \neq j\} | \{X_{ij} | i,j > |S|, i \neq j\})
    \end{equation}
    \begin{equation}
        =H(\{X_{ij} | i,j \leq |S|, i \neq j\}) = H(\graph_n[S])
    \end{equation}
    where the last line follows because no element in in $\{X_{ij} | i,j \leq |S|, i \neq j\}$ shares an edge with an element in $\{X_{ij} | i,j > |S|, i \neq j\}$, and so the collections are completely independent.
\end{proof}

\begin{lemma}[Neighbourhood AEPs]
    \label{thm:nbhd_aep}
    Let $\{S_n\}_{n=2}^\infty$ be a sequence of sets satisfying $S_n \subseteq V_n$ for each $n$. If $|S_n| \to \infty$ as $n\tti$ then we have
    \begin{equation}
        \plim_{n\tti} \frac{1}{\binom{|S_n|}{2}s(n)^d} \log \frac{1}{\prob{N_{S_n} | N_{S^c_{n}}}} = h^*
    \end{equation}
    Therefore, for every $\epsilon > 0$ there exist sets defined by
    \begin{equation}
        \mathcal{T}_n^{\epsilon}(S_n) := \left\{(N_{S_n}, N_{-{S_n}}) : \left|\frac{1}{\binom{|S_n|}{2}s(n)^d}\log \frac{1}{\prob{N_{S_n} |N_{-{S_n}}}} - h^*\right| < \epsilon\right\}
    \end{equation}
    for each $n \geq 2$, such that for any sequence $\{S_n\}_{n=2}^\infty$ with $S_n \subseteq V_n$ and $|S_n| \to \infty$, we have
    \begin{enumerate}
        \item $\prob{(N_{S_n},N_{{S_n^c}}) \in \mathcal{T}_n^\epsilon(S_n)} \to 1$
        \item $|\mathcal{T}_n^\epsilon(S_n)| \leq 2^{\binom{|S_n|}{2}s(n)^dh^*+\epsilon}$
    \end{enumerate}
    If $(N_{S_n},N_{S_n^c}) \in \mathcal{T}_n^\epsilon(S_n)$ we say $(N_{S_n},N_{S_n^c})$ is typical.
\end{lemma}
\begin{proof}
    Using the same logic (and notation) as in Lemma \ref{thm:nbhd_reduction}, 
    \begin{align}
        \prob{N_{S_n}|N_{S_n^c}} = \prob{\Gamma(S_n) | \Gamma(S_n^c)} \nonumber \\
        =\prob{\{X_{ij} | i,j \leq |S_n|, i \neq j\} | \{X_{ij} | i,j > |S_n|, i \neq j\})} \nonumber\\
        = \prob{\{X_{ij} | i,j \leq |S_n|, i \neq j\}} = \prob{G_n[S_n]}
    \end{align}
    due to the independence of the two sets. Now note that marginally, the distribution of $G_n[S_n]$ is that of an SRGG on node set $S_n$, with sparsity $s(n)$. Therefore, the statement in the Theorem becomes
    \begin{equation}
        \plim_{n\tti} \frac{1}{\binom{|S_n|}{2}s(n)^d}\log \frac{1}{\prob{G_n[S_n]}} = h^*
    \end{equation}
    which is directly true as a consequence of Theorem \ref{thm:aep}, provided that $|S_n| \to \infty$ as $n\tti$. The properties of the typical set follow from the same logic as the end of the proof of Theorem \ref{thm:aep}. 
\end{proof}

Equipped with these results, we are ready to prove Theorem \ref{thm:partial_dsc}.

\begin{proof}[Proof of Theorem \ref{thm:partial_dsc}]
    \textit{1) Achievability: } We can use the classical idea of random binning. Let $L$ be a fixed integer, then we design the codebook as follows for an $L$-tuple of rates $(R_1,...,R_L)$. For each $l = 1,...,L$, we generate assignments by independently assigning each $b_l^{(n)} \in \mathcal{B}_l^{(n)}$ a uniformly random element, which we call its index, from the codebook
    \begin{equation}
        \mathcal{M}_l^{(n)} := \{1,2,3,...,2^{\frac{n}{L}(n-1)s(n)^d R_l}\}
    \end{equation}
    we denote the index of $b_l^{(n)}$ as $\phi_l^{(n)}(b_l^{(n)})$. Note that more than one realisation $b_l^{(n)}$ may be assigned the same index. We reveal these assignments to both the encoder, and the decoder. To encode a graph, each encoder $l$ sees its block, $b_l^{(n)}$. Then, the independent encodings are sent as the tuple $(\phi_1^{(n)}(b_1^{(n)}),...,\phi_L^{(n)}(b_L^{(n)}))$. The decoder receives this tuple, and searches for their own tuple $(\hat{b}_1^{(n)},...,\hat{b}_L^{(n)})$ (or equivalently a graph $\hat{G}_n$), such that $(\hat{b}_1^{(n)},...,\hat{b}_L^{(n)})$ is unique, typical, and $\phi_l^{(n)}(\hat{b}_l^{(n)}) = \phi_l^{(n)}(b_l^{(n)})$ for every $l$. If they cannot find such a tuple, an error is declared. Therefore, an error can arise in one of two ways. First, the event
    \begin{equation}
        \mathcal{E}_1 = \{G_n \notin \mathcal{T}_n^\epsilon\} 
    \end{equation}
    is the event that the graph (or equivalently the $L$-tuple of blocks that the encoder compressed) is atypical. For a set $\Lambda \subseteq [L] := \{1,2,...,L\}$, let $B_{\Lambda}^{(n)} = \bigcup_{l\in \Lambda} B_l^{(n)}$. Then, with a slight abuse of notation, define the typical set $\mathcal{T}_n^\epsilon(\Lambda)$ to be the typical set in the statement of Theorem \ref{thm:nbhd_aep}, that is,
    \begin{equation}
        \mathcal{T}_n^\epsilon(\Lambda) = \mathcal{T}_n^\epsilon(\{i : (l-1)n/L < i \leq ln/L \text{ for some }l \in \Lambda\}).
    \end{equation}
    Then the second error event is
    \begin{equation}
        \mathcal{E}_2 := \bigcup_{\Lambda \subseteq [L]}\left\{\exists \hat{b}_\Lambda^{(n)} \neq B_\Lambda^{(n)} \space \text{ s.t. } (\hat{b}_\Lambda^{(n)}, B^{(n)}_{\Lambda^c}) \in \mathcal{T}_n^\epsilon(\Lambda) \text{ and } \phi_l^{(n)}(\hat{b}_l^{(n)}) = \phi_l^{(n)}(B_l^{(n)})\space \forall l \in \Lambda\right\}
    \end{equation}
    which is the probability that any subset of the blocks causes an error. For $\Lambda\subseteq [L]$ denote by $\mathcal{E}_2(\Lambda)$ the event corresponding to the error in the above union on the set $\Lambda$. Then by the union bound we have
    \begin{equation}
        P_E^{(n)} \leq \prob{\mathcal{E}_1} + \sum_{\Lambda \subseteq [L]} \prob{\mathcal{E}_2(\Lambda)}
    \end{equation}
    Clearly $\prob{\mathcal{E}_1} \to 0$ as a consequence of Theorem \ref{thm:aep}. Then since the sum is finite, we only need to demonstrate that the probability of error on an arbitrary subset goes to 0. Indeed, let $\Lambda \subseteq [L]$, then writing $\prob{\mathcal{E}_2(\Lambda)}$ as an expectation over the pair $(B^{(n)}_\Lambda, B^{(n)}_{\Lambda^c})$,
    \begin{align}
        \prob{\mathcal{E}_2(\Lambda)} = \mathbb{E}\big[\mathbb{P}\big(&\exists \hat{b}_\Lambda^{(n)}\neq B_\Lambda^{(n)} \space \text{ s.t. } (\hat{b}_\Lambda^{(n)}, B_{\Lambda^c}^{(n)}) \in \mathcal{T}_n^\epsilon(\Lambda) \nonumber \\ &\text{ and } \phi_l^{(n)}(\hat{b}_l^{(n)}) = \phi_l^{(n)}(B_l^{(n)})\space \forall l \in \Lambda\big)\big] \nonumber
    \end{align}
    \begin{align}
        = \mathbb{E}\big[\mathbb{P}\big(\exists \hat{b}_\Lambda^{(n)}\neq B_\Lambda^{(n)} \space \text{ s.t. } \phi_l^{(n)}(\hat{b}_l^{(n)}) = \phi_l^{(n)}(B_l^{(n)})\space \forall l \in \Lambda \big |& (\hat{b}_\Lambda^{(n)}, B_{\Lambda^c}^{(n)}) \in \mathcal{T}_n^\epsilon(\Lambda)\big)\nonumber \\ &\prob{(\hat{b}_\Lambda^{(n)}, B_{\Lambda^c}^{(n)}) \in \mathcal{T}_n^\epsilon(\Lambda)}\big] \nonumber
    \end{align}
    \begin{align}
        &\leq \mathbb{E}\big[\mathbb{P}\big(\exists \hat{b}_\Lambda^{(n)}\neq B_\Lambda^{(n)} \space \text{ s.t. } \phi_l^{(n)}(\hat{b}_l^{(n)}) = \phi_l^{(n)}(B_l^{(n)})\space \forall l \in \Lambda \big | (\hat{b}_\Lambda^{(n)}, B_{\Lambda^c}^{(n)}) \in \mathcal{T}_n^\epsilon(\Lambda)\big)\big]
    \end{align}
    \begin{align}
        = \sum_{B_{\Lambda}^{(n)}} \prob{B_{\Lambda}^{(n)}, B_{\Lambda^c}^{(n)}} \sum_{\substack{\hat{b}_\Lambda^{(n)} \neq B_{\Lambda}^{(n)} \\ (\hat{b}_\Lambda^{(n)}, B_{\Lambda^c}^{(n)}) \in \mathcal{T}_n^\epsilon(\Lambda)}} \prob{\phi_l^{(n)}(\hat{b}_l^{(n)}) = \phi_l^{(n)}(B_l^{(n)})\space \forall l \in \Lambda}
    \end{align}
    \begin{align}
        \leq \sum_{B_{\Lambda}^{(n)}} \prob{B_{\Lambda}^{(n)}, B_{\Lambda^c}^{(n)}}  2^{-\frac{n}{L}(n-1)s(n)^d \sum_{l\in \Lambda}R_l} |\mathcal{T}_n^\epsilon(\Lambda)|
    \end{align}
    \begin{equation}
        =2^{-\frac{n}{L}(n-1)s(n)^d \sum_{l\in \Lambda}R_l} |\mathcal{T}_n^\epsilon(\Lambda)|    
    \end{equation}
    Then, we can use the bound on the size of the typical set to give
    \begin{equation}
        2^{-\frac{n}{L}(n-1)s(n)^d \sum_{l\in \Lambda}R_l} |\mathcal{T}_n^\epsilon(\Lambda)|  \leq 2^{-\frac{n}{L}(n-1)s(n)^d \sum_{l\in \Lambda}R_l + \binom{v}{2}s(n)^d(h^*+\epsilon)}
    \end{equation}
    where $v = |\Lambda|n/L$ is the number of nodes contained in the blocks $l \in \Lambda$. This goes to 0 if
    \begin{equation}
        \frac{n}{L}(n-1)s(n)^d \sum_{l\in \Lambda}R_l > \binom{v}{2}s(n)^d(h^*+\epsilon)
    \end{equation}
    Dividing and taking the limit, we see that this occurs in the limit $n\tti$ if
    \begin{equation}
        \sum_{l \in \Lambda} R_l \geq \frac{|\Lambda|^2}{L}\frac{h^*}{2} 
    \end{equation}
    which is the condition in the statement of the Theorem. So, if a rate tuple $(\tilde{R}_1,...,\tilde{R}_L)$ satisfies the above condition for every $\Lambda \subseteq [L]$, then under the coding scheme described above we have $P_E^{(n)} \to 0$ as $n\tti$, and the rate tuple is achievable. \\

    \textit{2) Converse} For the converse direction, we use an extended version of \cite[Lemma 7.2.2]{han_book}, which easily follows by replacing the normalisation $1/n$ in their proof with our normalisation $\frac{1}{\frac{n}{L}(n-1)s\left(n\right)^d}$. This gives for every $\gamma > 0$,
    \begin{align}
        P_E^{(n)} \geq \mathbb{P}\bigg(\bigcup_{\Lambda \subseteq [L]} \bigg\{&\frac{1}{\frac{n}{L}(n-1)s\left(n\right)^d}\log \frac{1}{\prob{B_{\Lambda}^{(n)}| B_{\Lambda^c}^{(n)}}} \nonumber \\ &\geq \frac{1}{\frac{n}{L}(n-1)s\left(n\right)^d} \sum_{l\in \Lambda} \log |\mathcal{M}_l^{(n)}| + \gamma\bigg\}\bigg) - \sum_{\Lambda \subseteq [L]} 2^{-\frac{n}{L}(n-1)s\left(n\right)^d \gamma} 
    \end{align}
    So, in particular for any $\Lambda \subseteq [L]$,
    \begin{align}
        \label{eq:p_e_lower}
        P_E^{(n)} \geq \mathbb{P}\bigg(&\frac{1}{\frac{n}{L}(n-1)s\left(n\right)^d}\log \frac{1}{\prob{B_{\Lambda}^{(n)}| B_{\Lambda^c}^{(n)}}} \nonumber \\ &\geq \frac{1}{\frac{n}{L}(n-1)s\left(n\right)^d} \sum_{l\in \Lambda} \log |\mathcal{M}_l^{(n)}| + \gamma\bigg) - 2^{-\frac{n}{L}(n-1)s\left(n\right)^d \gamma} 
    \end{align}
    Now suppose that there exists a code with
    \begin{equation}
        \limsup_{n\tti} \frac{1}{\frac{n}{L}(n-1)s\left(n\right)^d}\sum_{l \in \Lambda} \log |\mathcal{M}_l^{(n)}| \leq \sum_{l\in \Lambda} R_l
    \end{equation}
    for some $\Lambda \subseteq [L]$ such that $P_E^{(n)} \to 0$. Then for any $\gamma >0$ we have
    \begin{equation}
        \frac{1}{\frac{n}{L}(n-1)s\left(n\right)^d}\sum_{l \in \Lambda} \log |\mathcal{M}_l^{(n)}| \leq \sum_{l \in \Lambda} R_l + \gamma
    \end{equation}
    for $n$ sufficiently large. Then \eqref{eq:p_e_lower} is further lower-bounded by
    \begin{equation}
        P_E^{(n)} \geq \mathbb{P}\Bigg(\frac{1}{\frac{n}{L}(n-1)s\left(n\right)^d}\log \frac{1}{\prob{B_{\Lambda}^{(n)}| B_{\Lambda^c}^{(n)}}} \geq \sum_{l \in \Lambda} R_l + \gamma\Bigg) - 2^{-\frac{n}{L}(n-1)s\left(n\right)^d \gamma}
    \end{equation}
    Taking the limit, we get
    \begin{equation}
        \lim_{n\tti} P_E^{(n)} \geq \lim_{n\tti} \mathbb{P}\Bigg(\frac{1}{\frac{n}{L}(n-1)s\left(n\right)^d}\log \frac{1}{\prob{B_{\Lambda}^{(n)}| B_{\Lambda^c}^{(n)}}} \geq \sum_{l \in \Lambda} R_l + \gamma\Bigg)
    \end{equation}
    \begin{equation}
        = \lim_{n\tti} \mathbb{P}\Bigg(\frac{1}{\binom{|\Lambda|n/L}{2}s\left(n\right)^d}\log \frac{1}{\prob{B_{\Lambda}^{(n)}| B_{\Lambda^c}^{(n)}}} \geq \frac{2}{|\Lambda|}\frac{s(n)^d(n-1)}{s\left(n\right)^d\left(\frac{|\Lambda|n}{L}-1\right)}\left(\sum_{l \in \Lambda} R_l + \gamma\right)\Bigg)
    \end{equation}
    \begin{equation}
        =\begin{cases}
            1 & \lim_{n\tti}\frac{2}{|\Lambda|}\frac{s(n)^d(n-1)}{s\left(n\right)^d\left(\frac{|\Lambda|n}{L}-1\right)}\left(\sum_{l \in \Lambda} R_l + \gamma\right) \leq h^* \\
            0 & \text{otherwise}
        \end{cases}
    \end{equation}
    where the last equality follows by the neighbourhood AEP Lemma \ref{thm:nbhd_aep}. Since $\gamma$ was arbitrary, this gives that the rate tuple must satisfy
    \begin{equation}
        \sum_{l \in \Lambda} R_l \geq \frac{|\Lambda|^2}{L}\frac{h^*}{2}
    \end{equation}
    for every subset $\Lambda \subseteq [L]$. Combining this observation with the achievability result completes the proof.
\end{proof}

\section{Numerical Results}

In this section, we plot the Slepian-Wolf rate region for a range of different connection functions. We define the following connection functions, as defined in \cite{dettmann2016random}. They are motivated by wireless communications applications, and all make use of parameters $r_0$, the typical connection range, and $\eta$, the path loss exponent.

\begin{enumerate}
    \item Rayleigh Single Input Single Output (SISO): $p(r) = \exp(-(r/r_0)^{\eta})$.
    \item Single Input Multiple Output / Multiple Input Single Output with $k$ inputs / outputs (SIMO/MISO ($k$)): $p(r)= \Gamma(k, (r/r_0)^\eta)/\Gamma(k)$.
    \item Multiple Input Multiple Output (MIMO (2,2)): $p(r) = \exp(-(r/r_0)^\eta)((r/r_0)^{2\eta} + 2 -\exp(-(r/r_0)^\eta))$.
    \item (Standard) Log-normal: $p(r) = \frac{1}{2}\text{erfc}\left(\frac{10}{\sqrt{2}}\log_{10}((r/r_0)^\eta)\right)$.
\end{enumerate}
Here $\Gamma(\cdot)$ is the Gamma function, $\Gamma(\cdot, \cdot)$ is the upper incomplete Gamma function, and $\text{erfc}(\cdot) =1-\text{erf($\cdot$)}$ is the complement of the error function. Figures \ref{fig:different_cf} and \ref{fig:rayleigh} show the Slepian-Wolf region for these connection functions.

\begin{figure}
    \centering
    \includegraphics[width=0.9\linewidth]{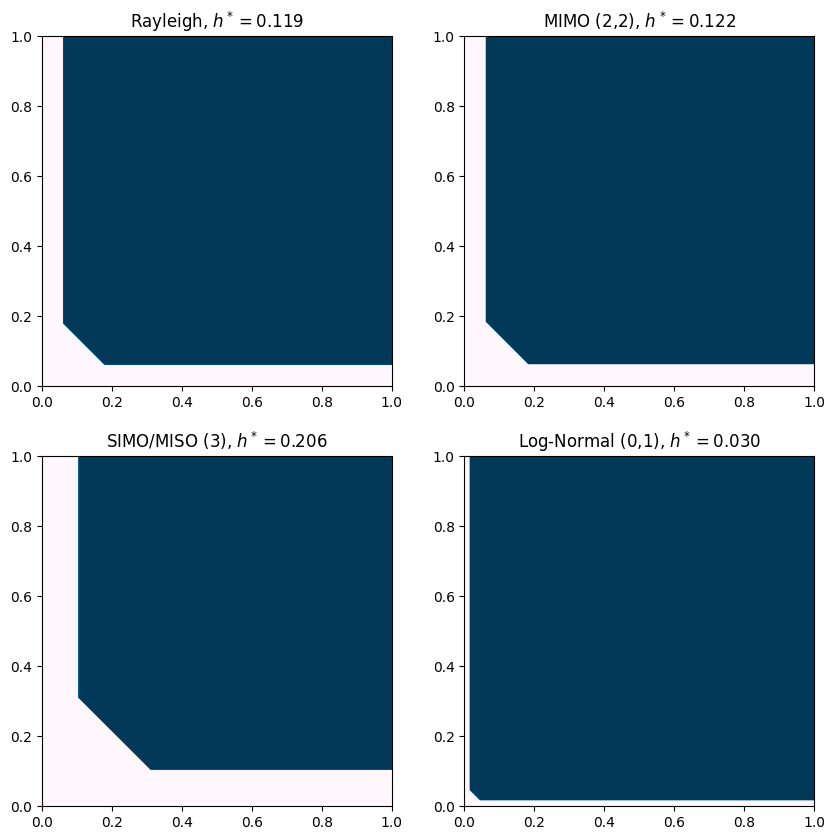}
    \caption{The Slepian-Wolf region for two encoders, with different connection functions. Here we use $r_0=0.1, \eta=2$ for every connection function.}
    \label{fig:different_cf}
\end{figure}

\begin{figure}
    \centering
    \includegraphics[width=0.9\linewidth]{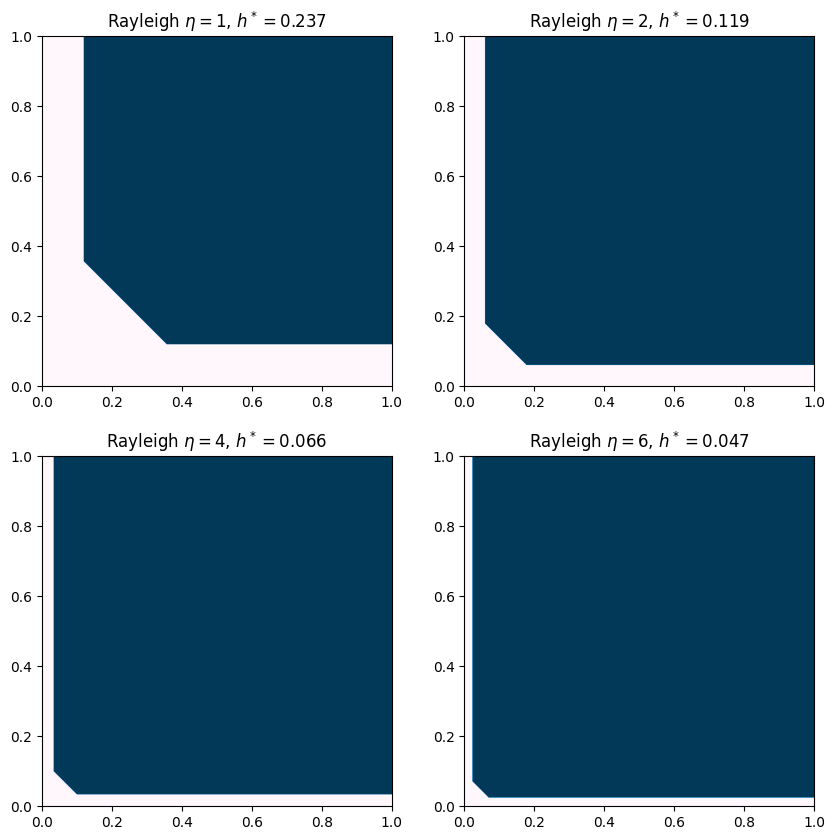}
    \caption{The Slepian-Wolf region for two encoders, using Rayleigh SISO, with $r_0=0.1$ fixed, and $\eta = 1,2,4,6$.}
    \label{fig:rayleigh}
\end{figure}

\section*{References}

\begingroup
    \renewcommand{\section}[2]{}

    \bibliographystyle{IEEEtran}
    \bibliography{supp}

\endgroup